\pdfoutput =1

\documentclass[pra,onecolumn,superscriptaddress]{article}

\usepackage[T1]{fontenc}
\usepackage[utf8]{inputenc}
\usepackage{amsmath,amsfonts,amsthm,amssymb}
\usepackage{mathtools}
\usepackage{subeqnarray}
\usepackage{setspace}
\usepackage{graphics,graphicx,color}
\usepackage{url}
\usepackage{enumerate}
\usepackage{float} 
\usepackage{multirow}
\usepackage{hhline}
\usepackage[margin=1.0in]{geometry}
\usepackage{braket}
\usepackage{dsfont} 
\usepackage{authblk}
\usepackage{multirow}
\usepackage{hhline}
\usepackage[makeroom]{cancel}
\usepackage{ifthen}
\usepackage{comment}
\usepackage{breakcites}

\usepackage{xcolor}
\definecolor{ddgreen}{rgb}{.05,.4,.05}
\definecolor{damethyst}{rgb}{0.4, 0.2, 0.6}
\usepackage[colorlinks, linkcolor=damethyst,citecolor=ddgreen]{hyperref}
\usepackage[capitalise]{cleveref}

\newtheorem{theorem}{Theorem}

\newtheorem{lem}{Lemma}

\newtheorem{conjecture}{Conjecture}

\newtheorem{remark}[theorem]{Remark}

\usepackage[scr=boondoxo,scrscaled=1.05]{mathalfa}

\newcommand{\Tr}[1]{\mbox{$\mathrm{Tr}\Big[#1\Big]$}}

\newcommand{\beq}{\begin{eqnarray}}
\newcommand{\Eeq}{\end{eqnarray}}

\usepackage{tikz}

\newcommand{\sfA}{\mathsf{A}}
\newcommand{\sfB}{\mathsf{B}}
\newcommand{\sfC}{\mathsf{C}}
\newcommand{\sfABC}{\mathsf{ABC}}

\DeclareMathOperator*{\Exp}{\mathbb{E}}

\definecolor{amethyst}{rgb}{0.6, 0.4, 0.8}

\title{The curious case of ``XOR repetition'' \\of monogamy-of-entanglement games
}
\author{
\makebox[\textwidth][c]{%
  Andrea Coladangelo\thanks{Paul G.\ Allen School of Computer Science \& Engineering, University of Washington. Email: \texttt{coladan@cs.washington.edu}} \quad
  Qipeng Liu\thanks{Computer Science and Engineering, UC San Diego. Email: \texttt{qipengliu@ucsd.edu}} \quad
  Ziyi Xie\thanks{Department of Computer Science and Technology, Tsinghua University. Email: \texttt{xie-zy21@mails.tsinghua.edu.cn}}
}}
\date{}
\begin{document}

\maketitle
\begin{abstract}
In this work, we consider ``decision'' variants of a well-known monogamy-of-entanglement game by Tomamichel, Fehr, Kaniewski, and Wehner [New Journal of Physics '13].
In its original ``search'' variant, Alice prepares a (possibly entangled) state on registers $\mathsf{ABC}$; register $\mathsf{A}$, consisting of $n$ qubits, is sent to a Referee, while $\mathsf{B}$ and $\mathsf{C}$ are sent to Bob and Charlie; the Referee then measures each qubit in the standard or Hadamard basis (chosen uniformly at random). The basis choices are sent to Bob and Charlie, whose goal is to simultaneously guess the Referee's $n$-bit measurement outcome string $x$. Tomamichel et al.\ show that the optimal winning probability is $\cos^{2n} {(\frac{\pi}{8})}$, following a perfect parallel repetition theorem. We consider the following ``decision'' variants of this game:
\begin{itemize}
\item Variant 1, ``XOR repetition'': Bob and Charlie's goal is to guess the XOR of all the bits of $x$. Ananth et al.\ [Asiacrypt '24] conjectured that the optimal advantage over random guessing decays exponentially in $n$. Surprisingly, we show that this conjecture is false, and, in fact, there is no decay at all: there exists a strategy that wins with probability $\cos^2{(\frac{\pi}{8})} \approx 0.85$ for any $n$. Moreover, this strategy does not involve any entanglement between Alice, Bob, and Charlie!
\item Variant 2, ``Goldreich-Levin'': The Referee additionally samples a uniformly random $n$-bit string $r$ that is sent to Bob and Charlie along with the basis choices. Their goal is to guess the parity of $r\cdot x$. We show that the optimal advantage over random guessing decays exponentially in $n$ for the restricted class of adversaries that do not share entanglement. A similar result was already shown by Champion et al.\ and {\c{C}}akan et al.; we give a more direct proof. 
\end{itemize}
Showing that Variant 2 is ``secure'' (i.e.,\ that the optimal winning probability is exponentially close to $\frac12$) against general adversaries would imply the existence of an information-theoretically ``unclonable bit''. We put forward a reasonably concrete conjecture that is equivalent to the general security of Variant 2.
\end{abstract}

\tableofcontents

\newpage 

\section{Introduction}
Monogamy of entanglement is a fundamental property of quantum information that captures the inability of a quantum system to be maximally entangled with two systems at once. This property is often studied through the framework of \emph{monogamy-of-entanglement games} (or monogamy games, for short), in which it manifests as a non-trivial upper bound on the winning probability. Monogamy games have found many applications throughout quantum cryptography, including some of the first proofs of security of quantum key distribution (via uncertainty relations)~\cite{koashi2006unconditional, tomamichel2017largely}, as well as in position verification~\cite{tomamichel2013monogamy}, unclonable encryption~\cite{broadbent2019uncloneable}, and copy-protection~\cite{coladangelo2021hidden, culf2022monogamy}.
In a monogamy game, Alice prepares a state $\rho_{\mathsf{ABC}}$, where register $\mathsf{A}$ consists of some fixed number $n$ of qubits, while $\mathsf{B}$ and $\mathsf{C}$ can be of arbitrary dimension; she sends $\mathsf{A}$ to the Referee, and $\mathsf{B}$ and $\mathsf{C}$ to her teammates Bob and Charlie respectively; the Referee performs a measurement on $\mathsf{A}$, sampled from a known family of possible measurements, obtaining some outcome $x$; the Referee reveals the measurement choice to both Bob and Charlie, whose goal is to simultaneously guess the measurement outcome $x$ (without communicating with each other).

In the simplest monogamy game, Alice sends a single qubit to the Referee, who measures it in either the standard or the Hadamard basis, chosen uniformly at random. It turns out that the optimal strategy for this game does not require any entanglement across the three registers: Alice simply sends the state $\cos{\frac{\pi}{8}}\ket{0} + \sin{\frac{\pi}{8}}\ket{1}$ to the Referee, and Bob and Charlie deterministically guess $0$. It is easy to see that this strategy achieves a winning probability of $\cos^2(\frac{\pi}{8})$.

In \cite{tomamichel2013monogamy}, Tomamichel et al.\ consider the parallel repetition of this game, i.e.\ $\mathsf{A}$ is an $n$-qubit register, the referee samples measurement bases independently and uniformly at random for each qubit, and Bob and Charlie are required to guess the entire $n$-bit outcome string. They show that the optimal winning probability follows a perfect parallel repetition theorem, i.e. it not only decays exponentially in $n$, but it is exactly $\cos^{2n}(\frac{\pi}{8})$. This result has found several applications throughout quantum cryptography.

Despite the importance of monogamy games, we still lack the answer to some seemingly basic questions. One example is the following. The monogamy game from Tomamichel et al.~\cite{tomamichel2013monogamy} can be thought of as a \emph{search}-based game, where the goal is to guess the entirety of the measurement outcome string $x$. Does a ``search-to-decision'' reduction exist? In other words, is there a way to transform a search-based monogamy game with exponentially small winning probability into a \emph{decision}-based game, where Alice and Bob need to guess only a \emph{single} bit, in such a way that the advantage is still exponentially small (over $\frac12$)? The answer to this question is not known. In fact, it is not known whether a decision-based game with this property exists at all. Thus, in this work, we are motivated by the following question:
\begin{center}
\emph{Does there exist a decision-based monogamy-of-entanglement game where the optimal\\winning probability is exponentially close to $\frac12$ (in the number of qubits $n$)?\\} 
\end{center}
This state of affairs is, in a way, made even more unsatisfying by the fact that we know of candidates that plausibly satisfy the desired property, but we are missing a proof. An affirmative answer to this question would immediately imply the existence of an ``unclonable bit'': an information-theoretically secure unclonable encryption scheme for a single bit. More precisely, the latter notion, put forward by Broadbent and Lord~\cite{broadbent2019uncloneable}, refers to a (private-key) encryption scheme with quantum ciphertexts that satisfies the following property: it is impossible for an adversary to process a ciphertext into two parts, each of which allows recovery of the plaintext given the secret key. Existing unclonable encryption schemes either satisfy a weaker search-based notion of security~\cite{broadbent2019uncloneable, ananth2021unclonable}, or are only proven secure in the random oracle model~\cite{ananth2022feasibility,ananth2023cloning}. Finding a provably-secure unclonable encryption scheme that satisfies the gold standard of ``indistinguishability''-based security in the plain model remains elusive. Botteron et al.~\cite{botteron2024towards} propose a candidate that they prove has exponentially decaying advantage up to certain constant values of $n$ based on the NPA hierarchy.
In exciting recent work~\cite{bhattacharyya2025uncloneable}, Bhattacharyya and Culf get close to the goal, by describing a scheme that provably achieves \emph{inverse-polynomial} indistinguishability-based security (just short of the desired negligible security). They do this by describing a ``decision''-based monogamy game where the optimal winning probability is inverse-polynomially close to $\frac12$. Given how fundamental the primitive of an ``unclonable bit'' appears to be, its connection to monogamy games, and the number of attempts in recent years, the problem of constructing it has arguably earned its place among the most interesting open questions in quantum cryptography.

\subsection{Our contributions}
We consider the following ``decision'' variants of the \cite{tomamichel2013monogamy} monogamy-of-entanglement game, and show the following results.

\noindent \emph{Variant 1} (``XOR repetition''): This is analogous to the game described at the start, except that Bob and Charlie's goal is to guess the \emph{XOR} of all the bits of $x$ (rather than the entire string $x$). A bit more precisely, Alice prepares a state $\rho_{\mathsf{ABC}}$, where register $\mathsf{A}$ consists of $n$ qubits; she sends $\mathsf{A}$ to the Referee, and $\mathsf{B}$ and $\mathsf{C}$ to Bob and Charlie respectively; the Referee measures each of the $n$ qubits in the standard or Hadamard basis (chosen independently and uniformly at random) - denote this choice by a string $\theta \in \{0,1\}^n$; Bob and Charlie, upon receiving $\theta$, must guess the \emph{XOR} of all the bits of $x$. 

The trivial strategy, where Bob and Charlie guess the same uniformly random bit, achieves a winning probability of $\frac12$. Ananth et al.~\cite{ananth2025unclonable} conjectured that the optimal advantage over $\frac12$ decays exponentially in $n$.\footnote{This was also the guess of one of the authors of \cite{tomamichel2013monogamy} in an email correspondence over five years ago.} Surprisingly, we show that there is no decay at all: there exists a strategy that wins with probability $\cos^2{(\frac{\pi}{8})} \approx 0.85$ for any $n$. Moreover, this strategy does not involve any entanglement between Alice, Bob, and Charlie! The strategy takes the following simple form: Alice sends a $n$-qubit pure state $\ket{\psi}$ to the Referee, and sends nothing to Bob and Charlie; upon receiving $\theta$, Bob and Charlie return a deterministic guess (which is a fixed function of $\theta$). We find this quite remarkable: there exists a fixed $n$-qubit state $\ket{\psi}$ (in fact, as we will see, an orthonormal basis of states of this kind) such that, no matter what the string of basis choices $\theta$ is, the XOR of the $n$ measurement outcomes is ``$\cos^2(\frac{\pi}{8})$-biased'' towards either $0$ or $1$ (which of the two values the XOR is biased towards may depend on $\theta$). 
\begin{theorem}[informal]
\label{thm:attack}
For any $n$, the optimal winning probability in the ``XOR repetition'' game (Variant 1) is $\cos^2(\frac{\pi}{8})$. Moreover, there is an optimal strategy where $\rho_{\mathsf{ABC}}$ is such that $\mathsf{A}$ is unentangled with $\mathsf{BC}$.
\end{theorem}

\noindent \emph{Variant 2} (``Goldreich-Levin style''): This is analogous to the above, except that the Referee additionally samples a uniformly random $n$-bit string $r$ that is sent to Bob and Charlie along with the basis choices. Their goal is to guess the parity of $r\cdot x$. 

We show that the optimal advantage over $\frac12$ decays exponentially for strategies that do not involve any entanglement between Alice, Bob, and Charlie. A similar result was already shown in~\cite{champion2024untelegraphable, ccakan2024unbounded}, but we give a more direct proof.\footnote{Technically, \cite{champion2024untelegraphable, ccakan2024unbounded} consider the closely-related notion of ''unclonability'' (or ''no-telegraphing'') games rather than monogamy games. A proof of security for a monogamy game implies security of the corresponding ''unclonability'' game. While it is possible the reverse is also true, this equivalence is not known.}
\begin{theorem}[informal]
For any $n$, the optimal winning probability in the ``Goldreich-Levin style'' game (Variant 2) is at most $\frac12 + 0.93^n$ for strategies where $\rho_{\mathsf{ABC}}$ is such that $\mathsf{A}$ is unentangled with $\mathsf{BC}$.
\end{theorem}
As a byproduct of our more direct analysis, we arrive at a somewhat concrete conjecture (Conjecture~\ref{conj:1}) that is equivalent to an exponential decay even against general strategies. For further details, see \Cref{sec:conj_imply_UE}.

\subsection*{Acknowledgements}
The authors thank Jiahui Liu, Er-Cheng Tang, Kabir Tomer, Dakshita Khurana, and James Hulett for helpful conversations.



\section{Notation}
\begin{itemize}
\item For strings $x,\theta \in \{0,1\}^n$, we denote
$\ket{x^{\theta}} = \bigotimes_{i \in [n]} \ket{x_i^{\theta_i}}$, where $\ket{0^{\theta_i}}$ and $\ket{1^{\theta_i}}$ are respectively the $+1$ and $-1$ eigenvectors of Pauli $Z$ (if $\theta_i = 0$) and of Pauli $X$ (if $\theta_i = 1$).

\item We write $\ket{0^Y} = \frac{1}{\sqrt{2}} \ket{0} + \frac{i}{\sqrt{2}} \ket{1}$ and $\ket{1^Y} =\frac{1}{\sqrt{2}}  \ket{0} -\frac{i}{\sqrt{2}} \ket{1}$. These are respectively the $\pm 1$ eigenvectors of Pauli $Y$. For a string $y \in \{0,1\}^n$, we use the notation 
$\ket{y^Y} = \ket{(y_1)^Y}\otimes \cdots \otimes \ket{(y_n)^Y}$. Moreover, we write $\bar{y}$ to denote the string obtained by flipping all of the bits of $y$.

\item For $x \in \{0,1\}^n$, we write $\textnormal{parity}(x)$ to denote the XOR of all the bits of $x$.
\end{itemize}

\section{The ``XOR repetition'' monogamy-of-entanglement game}

We consider the following game $\mathsf{XORMonogamy}(n)$, parametrized by some $n \in \mathbb{N}$.

\begin{itemize}
\item Alice prepares a state $\rho_{\sfABC}$ on registers $\sfABC$, where $\sfA$ is an $n$-qubit register, and $\sf{B}$ and $\sfC$ are of arbitrary dimension. Alice sends $\sfA$ to the referee, and $\sfB$ and $\sfC$ to Bob and Charlie respectively.
\item Referee samples $\theta \leftarrow \{0,1\}^n$, and measures in the basis $\{ \ket{x^{\theta}}\}_{x \in \{0,1\}^n}$.
\item Bob and Charlie both receive $\theta$. They also respectively receive registers $\sfB$ and $\sfC$. Bob and Charlie output bits $b$ and $b'$ respectively (by performing local operations on their registers).
\item The game is won if $b = b' = \textnormal{parity}(x)$, where the latter denotes the XOR of all the bits of $x$.
\end{itemize}

\paragraph{Optimal strategy for $n=1$.} When $n=1$, there is a well-known strategy achieving a winning probability of $\cos^2(\frac{\pi}{8}) \approx 0.85$: Alice sends to the referee the single-qubit state $\ket{\phi} = \cos (\frac{\pi}{8})\ket{0} + \sin(\frac{\pi}{8}) \ket{1}$, and nothing to Bob and Charlie (i.e.\ $\sfB, \sfC$ are empty). It is straightforward to see that, regardless of whether the Referee measures in the standard or Hadamard basis, the probability of outcome $0$ is $\cos^2(\frac{\pi}{8})$. Thus, Bob and Charlie can always guess $0$, and win with probability $\cos^2(\frac{\pi}{8})$. This strategy turns out to be optimal. 

There are three more optimal strategies that do not use entanglement, corresponding to the states $\ket{\psi} = \cos(\alpha) \ket{0} + \sin(\alpha) \ket{1}$ for $\alpha = \frac{3\pi}{8}, \frac{5\pi}{8}$, and $-\frac{\pi}{8}$. For each of these states, and for each choice of basis $\theta \in \{0,1\}$, there is one outcome that occurs with probability $\cos^2(\frac{\pi}{8})$. Thus, Bob and Charlie, who receive $\theta$, can again win with probability $\cos^2(\frac{\pi}{8})$ by guessing that outcome.

\subsection{Strategies achieving $\cos^2(\frac{\pi}{8})$ for general $n$}
\label{sec:attack}
In this section and the next, we prove Theorem~\ref{thm:attack}, namely that the optimal winning probability (over any quantum strategies) in the game $\mathsf{XORMonogamy}(n)$ is $\cos^2(\frac{\pi}{8})$ for any $n$, and that the optimal strategy is in fact semi-classical (i.e.\ Bob and Charlie are unentangled with Alice, and they answer deterministically). We start by giving an informal description of the optimal strategy, and why it achieves a winning probability of $\cos^2(\frac{\pi}{8})$.
We then describe the family of optimal strategies formally in Theorem~\ref{thm:attack-formal}, and provide a formal proof that this family achieves $\cos^2(\frac{\pi}{8})$. We argue optimality in Section~\ref{sec:optimality}.
 
One might naturally expect an exponential decay of the advantage over random guessing. Surprisingly, there is no decay at all: the optimal strategy achieves the same winning probability as in the $n=1$ case (i.e.\ $\cos^2(\frac{\pi}{8})$). The starting point to understand the optimal strategy is to notice first that, up to a global phase, the optimal single-qubit state $\cos{(\frac{\pi}{8})} \ket{0} + \sin{(\frac{\pi}{8})}\ket{1}$ is equal to the state $$ \frac{1}{\sqrt{2}} \ket{0^Y}+ \frac{1}{\sqrt{2}}e^{i\frac{\pi}{4}} \ket{1^Y} \,,
$$
where recall that $\ket{0^Y} = \frac{1}{\sqrt{2}} \ket{0} + \frac{i}{\sqrt{2}} \ket{1}$ and $\ket{1^Y} =\frac{1}{\sqrt{2}}  \ket{0} -\frac{i}{\sqrt{2}} \ket{1}$ are respectively the $\pm 1$ eigenvectors of Pauli $Y$. Now, consider the following $n$-qubit generalization of this state: 
$$\ket{\psi} = \frac{1}{\sqrt{2}} \ket{0^Y}^{\otimes n}+ \frac{1}{\sqrt{2}}e^{i\frac{\pi}{4}} \ket{1^Y}^{\otimes n}\,.$$
The claim is that $\ket{\psi}$ is optimal. Note that this means that $\ket{\psi}$ has quite a remarkable property: for all basis choices $\theta \in \{0,1\}^n$, the XOR of the $n$ measurement outcomes is ``$\cos^2(\frac{\pi}{8})$-biased'' towards either $0$ or $1$ (depending on $\theta$), i.e.\ there is one particular value that the XOR takes with probability $\cos^2(\frac{\pi}{8})$. 

A bit informally, the following is the reason why the claim is true. For any $\theta \in \{0,1\}^n$, let $X_{\theta}$ denote the observable corresponding to taking the XOR of the outcomes from a measurement in the basis $\theta$. Then, consider the two-dimensional subspace 
$$ S = \textnormal{span}\{\ket{0^Y}^{\otimes n}, \ket{1^Y}^{\otimes n} \} \,.$$
Observe crucially that $S$ is invariant under the action of $X_{\theta}$ (this is because the latter is a tensor product of single-qubit Pauli $X$ and $Z$ - so applying $X_\theta$ flips $\ket{0^Y}^{\otimes n}$ to $\ket{1^Y}^{\otimes n}$ and viceversa, up to a global phase). Now consider the isomorphism between $S$ and $\mathbb{C}^2$ such that $\ket{0^Y}^{\otimes n} \mapsto \ket{0^Y}$ and $\ket{1^Y}^{\otimes n} \mapsto \ket{1^Y}$. Then, this isomorphism must map $X_{\theta}$ to either $X$ or $Z$, since $X_{\theta}$ is the observable that flips $\ket{0^Y}^{\otimes n}$ to $\ket{1^Y}^{\otimes n}$, and vice versa, up to a phase. Thus, overall, the isomorphism maps $\ket{\psi}$ to the optimal single-qubit state $\cos (\frac{\pi}{8})\ket{0} + \sin(\frac{\pi}{8}) \ket{1}$, and it maps $X_{\theta}$ to either $X$ or $Z$, i.e.\ to either a standard or Hadamard basis measurement, both of which are $\cos^{2}(\frac{\pi}{8})$-biased. It turns out that there is not only one optimal state $\ket{\psi}$, but $2^{n+1}$ such states, which form \emph{two} orthonormal bases.

Recall that, for a string $y \in \{0,1\}^n$, we use the notation 
$\ket{y^Y} = \ket{(y_1)^Y}\otimes \cdots \otimes \ket{(y_n)^Y}$ (where $\ket{0^Y}$ and $\ket{1^Y}$ are the eigenvectors of Pauli $Y$), and we write $\bar{y}$ to denote the string obtained by flipping all of the bits of $y$. Moreover, recall that $\textnormal{parity}(x)$ denotes the XOR of all the bits of $x$.

\begin{theorem}[Optimal strategies for the XOR monogamy-of-entanglement game]
    \label{thm:attack-formal} 
    For any $y \in \{0,1\}^n$, let $$\ket{\varphi_y}=\frac{1}{\sqrt{2}}\ket{y^Y} + \frac{1}{\sqrt{2}}e^{i\frac{\pi}{4}}\ket{\bar{y}^Y}\,,$$
where the notation is defined right before this theorem.
Then, for all $\theta \in \{0,1\}^n$, the following holds: with probability $\cos^{2}(\frac{\pi}{8}) \approx 0.853$ over performing a measurement of $\ket{\varphi_y}$ in the $\theta$ basis, the outcome string $x$ is such that $\textnormal{parity}(x)$ is
    \begin{equation}
    \label{eq:parity_outcome}
        \bigg\{\begin{array}{ccc}
            y\cdot \theta &\mod{2} & \text{if } |\theta| = 0,1 \mod{4}\\
            y\cdot \theta + 1 &\mod{2} & \text{if } |\theta| = 2,3 \mod{4}
        \end{array}.
    \end{equation}
Moreover, let 
$$\ket{\psi_y}=\frac{1}{\sqrt{2}}\ket{y^Y} - \frac{1}{\sqrt{2}}e^{i\frac{\pi}{4}}\ket{\bar{y}^Y} \,.$$
Then, for all $\theta \in \{0,1\}^n$, the following holds: with probability $\cos^{2}(\frac{\pi}{8}) \approx 0.853$ over performing a measurement of $\ket{\varphi_y}$ in the $\theta$ basis, the outcome string $x$ is such that $\textnormal{parity}(x)$ is
\begin{equation}
        \bigg\{\begin{array}{ccc}
            y\cdot \theta &\mod{2} & \text{if } |\theta| = 0,3 \mod{4}\\
            y\cdot \theta + 1 &\mod{2} & \text{if } |\theta| = 1,2 \mod{4}
        \end{array}.
    \end{equation}
\end{theorem}
\begin{proof}
We will show first that the above holds for the states $\ket{\varphi_{y}} = \frac{1}{\sqrt{2}}\ket{y^Y} + \frac{1}{\sqrt{2}}e^{i\frac{\pi}{4}}\ket{\bar{y}^Y}$. The proof for the states $\ket{\psi_y}=\frac{1}{\sqrt{2}}\ket{y^Y} - \frac{1}{\sqrt{2}}e^{i\frac{\pi}{4}}\ket{\bar{y}^Y}$ is similar.

Define $\widetilde{Y} = (-1)^{y_1} Y \otimes I^{\otimes n-1}$, where $y_1$ is the first bit of $y$. Then, for any $\theta \in \{0,1\}^{n}$, define $\widetilde{X}_{\theta} = W_{\theta_1} \otimes \cdots \otimes W_{\theta_n}$, where $$W_{\theta_i} = \begin{cases}
Z \textnormal{ if } \theta_i = 0 \\
X \textnormal{ if } \theta_i = 1
\end{cases}\,.$$

Notice that observable $X_{\theta}$ ``measures'' exactly the parity of the outcome string resulting from a measurement in the basis $\theta$. This is because the eigenvectors of $X_{\theta}$ are precisely the states $\ket{x^{\theta}}$, with eigenvalue $+1$ if $\textnormal{parity}(x)=0$, and $-1$ if $\textnormal{parity}(x)=1$. Thus, for any state $\ket{\psi}$, 
$\Tr{X_{\theta} \ket{\psi}\bra{\psi}}$ is the expectation of $\textnormal{parity}(x)$, where $x$ is the outcome obtained from measuring $\ket{\psi}$ in the basis $\theta$.

Now, 
consider the two-dimensional subspace 
$$ S = \textnormal{span}\{\ket{y^Y}, \ket{\bar{y}^Y} \} \,.$$
Notice that $S$ is invariant under both $\widetilde{Y}$ and $\widetilde{X}_{\theta}$. This is because $\ket{y^Y}$, $\ket{\bar{y}^Y}$ are respectively $+1$ and $-1$ eigenvectors of $\widetilde{Y}$, and moreover 
\begin{align}
\widetilde{X}_{\theta} \ket{y^Y} &= (-1)^{y\cdot \theta} \cdot i^{|\theta|} \ket{\bar{y}^Y} \label{eq:101}\\
\widetilde{X}_{\theta} \ket{\bar{y}^Y} &= (-1)^{\bar{y}\cdot \theta} \cdot i^{|\theta|}  \ket{y^Y} \,, \label{eq:102}
\end{align}
where the phase arises from the different action that the (single-qubit) $Z$ and $X$ have on the eigenvectors of~$Y$.

We now consider the restrictions of $\widetilde{Y}$ and $\widetilde{X}_{\theta}$ to subspace $S$. Denote these by $\widetilde{Y}|_S$ and $\widetilde{X}_{\theta}|_S$. Since $\ket{\varphi_{y}} = \frac{1}{\sqrt{2}}\ket{y^Y} + \frac{1}{\sqrt{2}}e^{i\frac{\pi}{4}}\ket{\bar{y}^Y} \in S$, we have
\begin{equation}
\Tr{\widetilde{X}_{\theta} \ket{\varphi_y}\bra{\varphi_y}} =  \Tr{\widetilde{X}_{\theta}|_S \ket{\varphi_y}\bra{\varphi_y}} \,. \label{eq:1000}
\end{equation}
Consider the isomorphism $U: S \rightarrow \mathbb{C}^2$ such that $U \ket{y^Y} = \ket{0^Y}$ and $U \ket{\bar{y}^Y} = \ket{1^Y}$ (this is an isomorphism since it preserves inner products). Then, \begin{align}
U \ket{\varphi_y} &= \frac{1}{\sqrt{2}} \ket{0^Y} + \frac{1}{\sqrt{2}} e^{i\frac{\pi}{4}} \ket{1^Y}, \textnormal{ and} \label{eq:1001}\\
U \widetilde{Y}|_S U^{\dagger} &= Y  \,, \label{eq:1002}
\end{align}
where the latter follows from the fact that, by the definition of $\tilde{Y}$, $\ket{y^Y}$ and $\ket{\bar{y}^Y}$ are respectively the $+1$ and $-1$ eigenvectors of $\tilde{Y}|_S$. 
Now, what is $U \widetilde{X}_{\theta}|_S U^{\dagger}$? This depends on the parity of $y\cdot \theta$ and the value of $|\theta|\mod{4}$.
\begin{itemize}
\item When $|\theta| = 0 \mod{4}$, Equations \eqref{eq:101} and \eqref{eq:102} become
\begin{align*}
\widetilde{X}_{\theta} \ket{y^Y} &= (-1)^{y\cdot \theta}\ket{\bar{y}^Y} \\
\widetilde{X}_{\theta} \ket{\bar{y}^Y} &= (-1)^{\bar{y}\cdot \theta} \ket{y^Y} =(-1)^{y\cdot \theta} \ket{y^Y}   \,,
\end{align*}
where the last equality is because $\bar{y} \cdot \theta = y\cdot \theta + 1^n \cdot \theta = y \cdot \theta$, since $|\theta|= 0 \mod 4$ (and hence also mod~2).
This implies that $U \widetilde{X}_{\theta}|_S U^{\dagger} = (-1)^{y\cdot \theta} Z$, since $Z$ is the unique linear operator mapping $\ket{0^Y}$ to $\ket{1^Y}$, and $\ket{1^Y}$ to $\ket{0^Y}$. Combining this with \eqref{eq:1000} and \eqref{eq:1001}, this implies
\begin{align}
\Tr{\widetilde{X}_{\theta} \ket{\varphi_{y}}\bra{\varphi_{y}}} &=  \Tr{\widetilde{X}_{\theta}|_S \ket{\varphi_{y}}\bra{\varphi_{y}}} \\
&= (-1)^{y\cdot \theta}\Tr{Z \ket{\phi}\bra{\phi}} \,,
\end{align}
where $\ket{\phi} = \frac{1}{\sqrt{2}} \ket{0^Y} + \frac{1}{\sqrt{2}} e^{i\frac{\pi}{4}} \ket{1^Y}$. Thus, measuring observable 
on $\widetilde{X}_{\theta}$ on $\ket{\varphi_{y}}$ gives identical statistics as measuring observable $(-1)^{y\cdot \theta}Z$ on $\ket{\phi}$. Crucially, notice that, up to a global phase, 
$\ket{\phi} = \cos{\frac{\pi}{8}} \ket{0} + \sin{\frac{\pi}{8}} \ket{1}$, which is the state that yields an optimal attack in the single-qubit monogamy game. Thus, when performing a measurement in the basis $\theta$ on state $\ket{\varphi_{y}}$, the probability that $\textnormal{parity}(x) = 0$, where $x$ is the measurement outcome, is precisely
$$ \frac{1 + \Tr{\widetilde{X}_{\theta} \ket{\varphi_{y}}\bra{\varphi_{y}}}}{2} = \frac{1 + (-1)^{y\cdot \theta}\Tr{Z \ket{\phi}\bra{\phi}} }{2} \,.$$
The latter is equal to $\cos^2{\frac{\pi}{8}}$ when $y\cdot \theta = 0 \mod{2}$, and $1-\cos^2{\frac{\pi}{8}}$ when $y\cdot \theta = 1 \mod{2}$.

\item When $|\theta| = 1 \mod{4}$, Equations \eqref{eq:101} and \eqref{eq:102} become \begin{align*}
\widetilde{X}_{\theta} \ket{y^Y} &= (-1)^{y \cdot \theta} i  \ket{\bar{y}^Y} \\
\widetilde{X}_{\theta} \ket{\bar{y}^Y} &= -(-1)^{y \cdot \theta} i \ket{y^Y} \,.
\end{align*}
This implies that $U \widetilde{X}_{\theta}|_S U^{\dagger} = (-1)^{y \cdot \theta} X$. Thus, the probability that $\textnormal{parity}(x) = 0$ is 
$$ \frac{1 + \Tr{\widetilde{X}_{\theta} \ket{\varphi_{y}}\bra{\varphi_{y}}}}{2} = \frac{1 + (-1)^{y \cdot \theta}\Tr{X \ket{\phi}\bra{\phi}} }{2} \,.$$
The latter is equal to $\cos^2{\frac{\pi}{8}}$ when $y\cdot \theta = 0 \mod{2}$, and $1-\cos^2{\frac{\pi}{8}}$ when $y\cdot \theta = 1 \mod{2}$.
\item When $|\theta| = 2 \mod{4}$, Equations \eqref{eq:101} and \eqref{eq:102} become \begin{align*}
\widetilde{X}_{\theta} \ket{y^Y} &= (-1)^{y\cdot \theta + 1}\ket{\bar{y}^Y} \\
\widetilde{X}_{\theta} \ket{\bar{y}^Y} &= (-1)^{y\cdot \theta + 1} \ket{y^Y}   \,.
\end{align*}
This implies that $U \widetilde{X}_{\theta}|_S U^{\dagger} = (-1)^{y\cdot \theta + 1}Z$. Thus, the probability that $\textnormal{parity}(x) = 0$ is
$$ \frac{1 + \Tr{\widetilde{X}_{\theta} \ket{\varphi_{y}}\bra{\varphi_{y}}}}{2} = \frac{1 + (-1)^{y \cdot \theta + 1}\Tr{Z \ket{\phi}\bra{\phi}} }{2} \,.$$
The latter is equal to $\cos^2{\frac{\pi}{8}}$ when $y\cdot \theta = 1 \mod{2}$, and $1-\cos^2{\frac{\pi}{8}}$ when $y\cdot \theta = 0 \mod{2}$.
\item When $|\theta| = 3 \mod{4}$, Equations \eqref{eq:101} and \eqref{eq:102} become \begin{align*}
\widetilde{X}_{\theta} \ket{y^Y} &= (-1)^{y \cdot \theta + 1} i  \ket{\bar{y}^Y} \\
\widetilde{X}_{\theta} \ket{\bar{y}^Y} &= -(-1)^{y \cdot \theta + 1} i \ket{y^Y} \,.
\end{align*}
This implies that $U \widetilde{X}_{\theta}|_S U^{\dagger} = (-1)^{y\cdot \theta + 1}X$. Thus, the probability that $\textnormal{parity}(x) = 0$ is
$$ \frac{1 + \Tr{\widetilde{X}_{\theta} \ket{\varphi_{y}}\bra{\varphi_{y}}}}{2} = \frac{1 + (-1)^{y\cdot \theta + 1}\Tr{X \ket{\phi}\bra{\phi}} }{2}  \,.$$
The latter is equal to $\cos^2{\frac{\pi}{8}}$ when $y\cdot \theta = 1 \mod{2}$, and $1-\cos^2{\frac{\pi}{8}}$ when $y\cdot \theta = 0 \mod{2}$.
\end{itemize}
Putting things together, since the above holds for any $\theta,y \in \{0,1\}^n$, we get the statement of the theorem. Note that in this argument it is crucial that the operator $\widetilde{Y}$ is \emph{independent} of $\theta$, and so it is the same state $\ket{\varphi_{y}}$ that satisfies the desired property for all $\theta$.

The statement for the states $\ket{\psi_y}=\frac{1}{\sqrt{2}}\ket{y^Y} - \frac{1}{\sqrt{2}}e^{i\frac{\pi}{4}}\ket{\bar{y}^Y}$ is proven similarly, with the difference that the same isomorphism $U$ (defined right before Equation~\ref{eq:1001}) now maps $\ket{\psi_y}$ to $\frac{1}{\sqrt{2}} \ket{0^Y} - \frac{1}{\sqrt{2}} e^{i\frac{\pi}{4}} \ket{1^Y}$, which is equal to $\sin{\frac{\pi}{8}} \ket{0} - \cos{\frac{\pi}{8}} \ket{1}$ up to a global phase. The latter is another optimal state for a single-qubit strategy.
\end{proof}
Notice that the $2^{n+1}$ states in $\{\ket{\phi_y}\}$ and $\{\ket{\psi_y}\}$ form \emph{two} orthonormal bases, each comprising half of the states from each family.

\begin{remark}
Note that the families of strategies described here cannot be lifted to the game where the challenger additionally samples a uniformly random $r \in \{0,1\}^n$ (which is revealed to Bob and Charlie together with $\theta$), and Bob and Charlie have to guess $r \cdot \theta$. The step in the above argument that one cannot get to work is to argue that the subspace $S$ is invariant under $\widetilde{X}_{\theta}$. This is because one now needs to consider the operator $\widetilde{X}_{\theta, r}$, defined analogously to $\widetilde{X}_{\theta}$ except that it acts trivially at all locations where $r$ is $0$. However, notice that now $\widetilde{X}_{\theta, r} \ket{y^Y} = \ket{(y\oplus r)^Y}$, up to a global phase. Since $r$ can be any string, one cannot find a fixed 2-dimensional subspace that is invariant under all $\widetilde{X}_{\theta, r}$ and $\widetilde{Y}$. In fact, as we prove in Theorem~\ref{thm:main}, the optimal winning probability at this game is exponentially close to $\frac12$ (for semi-classical strategies).
\end{remark}

\begin{remark}
There is a symmetry that allows one to derive the ``optimal states'' $\{\ket{\phi_y}\}$ and $\{\ket{\psi_y}\}$ assuming one just knows that $\ket{\varphi_{0^n}}$ is optimal. The symmetry is the following: for any $a,b \in \{0,1\}^n$, the set of states $\{\ket{x^{\theta}} : x, \theta\in \{0,1\}^n\}$ is invariant, up to global phases, under the action of $X^a Z^b$. Thus, since the state $\ket{\varphi_{0^n}}$ satisfies the property of Theorem~\ref{thm:attack-formal}, 
$X^aZ^b\ket{\varphi_{0^n}}$ must also satisfy it, up to modifying the function that predicts the most likely parity outcome. A straightforward calculation then shows that the most likely parity outcome is as in \eqref{eq:parity_outcome}.
\end{remark}

\begin{remark}
The $2^{n+1}$ optimal states $\{\ket{\phi_y}\}$ and $\{\ket{\psi_y}\}$ are actually the only ones. While we do not provide a proof here, this can be shown with some work starting from the expressions derived in Sections~\ref{sec:analysis-first-part} and~\ref{sec:analysis-unentangled} (specialized to the case of $r=1^n$ and semi-classical strategies).
\end{remark}

\subsection{Proof of optimality}
\label{sec:optimality}
The fact that no strategy (even fully quantum) can achieve a higher winning probability than $\cos^2(\frac{\pi}{8})$ in $\mathsf{XORMonogamy}(n)$, for any $n$, follows straightforwardly from the fact that $\cos^2(\frac{\pi}{8})$ is optimal for $\mathsf{XORMonogamy}(1)$, i.e.\ the original game of \cite{tomamichel2013monogamy} when $n=1$ (note that when $n=1$, guessing the parity of a string is equivalent to guessing the entire string so the two games coincide). 

Formally, suppose for a contradiction there was a strategy $S$ for $\mathsf{XORMonogamy}(n)$ winning with probability $p > \cos^2(\frac{\pi}{8})$. Let $\rho_{\mathsf{ABC}}$ be the state used by this strategy, where $\mathsf{A}$ is an $n$-qubit register. We will construct a strategy for $\mathsf{XORMonogamy}(1)$ winning with the same probability $p$. The strategy is simple. In short, Alice creates the state $\rho_{\mathsf{ABC}}$ for strategy $S$, and sends only the first qubit of $\mathsf{A}$ to the referee; she simulates the actions of the referee on the remaining $n-1$ qubits. More formally:
\begin{itemize}
\item Alice creates $\rho_{ABC}$ (the state used by strategy $S$ for $\mathsf{XORMonogamy}(n)$, where $\mathsf{A}$ is an $n$-qubit register). She sends only the first qubit $\mathsf{A_1}$ to the referee (she keeps the other $n-1$ qubits $A_{2\cdots N}$). For $i \in \{2,\ldots, n\}$, Alice samples $\theta_i \leftarrow \{0,1\}$ uniformly at random and measures register $\mathsf{A_i}$ in the basis corresponding to $\theta_i$. Denote the outcome by $x_i$. Alice sends the concatenations $\theta_2 || \cdots || \theta_n$ and $x_2 || \cdots || x_n$ to both Bob and Charlie, along with registers $\mathsf{B}$ to Bob, and $\mathsf{C}$ to Charlie.
\item The referee for $\mathsf{XORMonogamy}(1)$ samples $\theta_1 \leftarrow \{0,1\}$ uniformly at random, and measures register $\mathsf{A_1}$ in the basis corresponding to $\theta_1$. Denote the outcome by $x_1$.
\item Bob and Charlie receive $\theta_1$ from the referee (as well as $\theta_2 || \cdots || \theta_n$, $x_2 || \cdots || x_n$, and registers $\mathsf{B}$ and $\mathsf{C}$ respectively from Alice). Let $\theta = \theta_1 || \cdots || \theta_n$. Bob measures register $\mathsf{B}$ according to ``Bob'' in strategy $S$ on input $\theta$. Let $b$ be the outcome. Bob returns $\widetilde{b} = b \oplus x_2 \oplus \cdots \oplus x_n$. Similarly, Charlie measures register $\mathsf{C}$ according to ``Charlie'' in strategy $S$ on input $\theta$. Let $b'$ be the outcome. Charlie returns $\widetilde{b}' = b' \oplus x_2 \oplus \cdots \oplus x_n$. 
\end{itemize}
Since strategy $S$ succeeds with probability $p > \cos^2(\frac{\pi}{8})$ in $\mathsf{XORMonogamy}(n)$, the strategy above is such that $b = b' = x_1 \oplus \cdots \oplus x_n$ with probability $p$. This implies that we have $\widetilde{b} = \widetilde{b}' = x_1$ with probability $p$, which is a contradiction.

\section{The ``Goldreich-Levin style'' monogamy-of-entanglement game}
\label{sec:1}
We consider the following game $\mathsf{GLMonogamy}(n)$, parametrized by some $n \in \mathbb{N}$.

\begin{itemize}
\item Alice prepares a state $\rho_{\sfABC}$ on registers $\sfABC$, where $\sfA$ is an $n$-qubit register, and $\sf{B}$ and $\sfC$ are of arbitrary dimension. Alice sends $\sfA$ to the referee. 
\item Referee samples $\theta, r \leftarrow \{0,1\}^n$, and measures in basis $\{ \ket{x^{\theta}}\}_{x \in \{0,1\}^n}$.
\item Bob and Charlie both receive $\theta$ and $r$. They also respectively receive registers $\sfB$ and $\sfC$. Bob and Charlie output bits $b$ and $b'$ respectively.
\item The game is won if $b = b' = r\cdot x \mod{2}$.
\end{itemize}
We now describe the general form of a strategy in $\mathsf{GLMonogamy}(n)$, as well as corresponding expressions for the optimal winning probability.

\paragraph{Fully quantum strategies.} A fully quantum strategy for Alice, Bob, and Charlie is specified by a state $\ket{v}_{\sfABC}$ (this can be taken to be pure without loss of generality), and families of projective measurements 
$$ \Big\{ \{P_0^{\theta, r}, P_1^{\theta, r} \} : \theta, r \in \{0,1\}^n \Big \}$$
on register $\sfB$, and
$$ \Big\{ \{Q_0^{\theta, r}, Q_1^{\theta, r} \} : \theta, r \in \{0,1\}^n \Big \}$$
on register $\sfC$, 
where $\{P_0^{\theta, r}, P_1^{\theta, r} \}$ is the measurement that Bob performs on $\sfB$ upon receiving $\theta, r$, and similarly $\{Q_0^{\theta, r}, Q_1^{\theta, r} \}$ for Charlie on $\sfC$.

\paragraph{Optimal winning probability.} For $\theta, r \in \{0,1\}^n$, define the projection
\begin{equation}
\label{eq:pi}
\Pi^{\theta,r} = \sum_{b} \sum_{x: r\cdot x = b} \ket{x^{\theta}}\bra{x^{\theta}} \otimes P_b^{\theta, r} \otimes Q_b^{\theta, r} \,. 
\end{equation}
Note that, for a strategy specified by a state $\ket{v}_{\sfABC}$, and projections $\{P^{\theta, r}_b\}$ on $\sfB$, and $\{Q^{\theta, r}_b\}$ on $\sfC$, its winning probability is precisely
$$ \Exp_{\theta, r} \| \Pi^{\theta, r} \ket{v}  \|^2  = \bra{v} \Exp_{\theta, r} \Pi^{\theta, r} \ket{v}   \,. $$
Thus, the optimal winning probability in $\mathsf{GLMonogamy}(n)$ is 
\begin{equation}
\label{eq:p_opt}
p_{opt}(n) = \max_{\ket{v}, \{P_{b}^{\theta, r}\}, \{Q_b^{\theta, r} \}} \bra{v} \Exp_{\theta, r} \Pi^{\theta, r} \ket{v} 
= \max_{\{P_{b}^{\theta, r}\}, \{Q_b^{\theta, r} \}} \max_{\ket{v}}   \bra{v} \Exp_{\theta, r} \Pi^{\theta, r} \ket{v}  \,.
\end{equation}

We will at first restrict our attention to ``semi-classical'' strategies, defined as follows.

\paragraph{``Semi-classical'' strategies.} 
Using the notation introduced above to describe general strategies, we say that a strategy is ``semi-classical'' if $\ket{v}_\sfABC = \ket{u}_{\sfA} \otimes \ket{w}_{\sfB\sfC}$, for some $\ket{u}_{\sfA}$ and $\ket{w}_{\sfB\sfC}$, and, moreover, for all $\theta, r$, it holds that
$P^{\theta, r}_0, P^{\theta,r}_1 \in \{0_\sfB,I_\sfB\}$ and $Q^{\theta,r}_0, Q^{\theta,r}_1 \in \{0_\sfC,I_\sfC\}$. In words, Bob and Charlie both return the output of a deterministic function of $\theta$ and $r$. Without loss of generality, we can assume that they answer according to the \emph{same} deterministic function of $\theta$ and $r$ (this is because a strategy where Alice and Bob answer using different functions performs at most as well as the strategy where they both answer according to Alice's function). Therefore, for any ``semi-classical'' strategy, there exists some function $c$ such that 
$$P_b^{\theta,r} = Q_b^{\theta,r} = \delta_{b = c(\theta,r)} \cdot I\,.$$
In this note, we prove the following. 
\begin{theorem}
\label{thm:main}
The optimal winning probability in $\mathsf{GLMonogamy}(n)$ for semi-classical strategies is at most 
$$ \frac12 + 0.93^n \,.$$
\end{theorem}
This theorem is proved in Section \ref{sec:security}.

\section{Security against ``semi-classical'' strategies}
\label{sec:security}
The first part of the proof (\cref{sec:analysis-first-part}) applies to general strategies. In particular, Conjecture \ref{conj:1}, which we state at the end of Subsection \ref{sec:analysis-first-part}, is equivalent to the winning probability in $\mathsf{GLMonogamy}(n)$ being exponentially close to $\frac12$ for all fully quantum strategies.

In the second part (\cref{sec:analysis-unentangled}), we restrict our attention to ``semi-classical'' strategies, and we complete the proof of \cref{thm:main}. The distinction between the game $\mathsf{XORMonogamy}(n)$ (whose optimal winning is $\cos^2(\frac{\pi}{8})$ from Theorem~\ref{thm:attack-formal}) and  $\mathsf{GLMonogamy}(n)$ becomes clear here, as it enables the use of Parseval's identity at the right time to bound the operator norm of a certain operator.

\subsection{Rewriting the optimal winning probability}
\label{sec:analysis-first-part}
From \cref{eq:p_opt}, our goal is to bound 
\begin{equation}
\label{eq:popt}
p_{opt}(n) = \max_{\{P_{b}^{\theta, r}\}, \{Q_b^{\theta, r} \}} \max_{\ket{v}}   \bra{v} \Exp_{\theta, r} \Pi^{\theta, r} \ket{v} \,,
\end{equation}
where $\Pi^{\theta,r} = \sum_{b} \sum_{x: r\cdot x = b} \ket{x^{\theta}}\bra{x^{\theta}} \otimes P_b^{\theta, r} \otimes Q_b^{\theta, r}$, and $\ket{v}$ is a normalized state.

For the rest of this subsection, we will fix some $\{P^{\theta,r}_b\}$ and $\{Q^{\theta,r}_b\}$. Let $\ket{v} = \sum_{i,k} \alpha_{ik} \ket{i}_{\sfA}\ket{k}_{\sfB \sfC}$. Let $\gamma_v = \bra{v} \Exp_{\theta,r} \Pi^{\theta,r} \ket{v}$. Then, 
$$ \gamma_v = \sum_{i,k,j,l} \overline{{\alpha}_{ik}} \cdot \alpha_{jl} \cdot \gamma_{ijkl} \,,$$
where $\gamma_{ijkl} = \bra{i}_{\sfA}\bra{k}_{\sfB \sfC} \Exp_{\theta,r} \Pi^{\theta, r}\ket{j}_{\sfA}\ket{l}_{\sfB \sfC} \,.$
We now proceed to compute simpler expressions for the $\gamma_{ijkl}$. We first introduce some notation.
\paragraph{Notation.} For an operator $A$, we use write $\| A \|$ to denote its operator norm. For a given $\theta \in \{0,1\}^n$, we use the notation $C = \{i: \theta_i = 0\}$, and $H = \{i: \theta_i = 1\}$. For a string $x \in \{0,1\}^n$, and a subset $S \subseteq [n]$, we write $x_S$ to denote the sub-string of $x$ corresponding to indices in $C$.

\vspace{2mm}
Fix some $i,j,k,l$. We have
\begin{align}
\gamma_{ijkl} &= \Exp_{\theta, r} \bra{i}\bra{k} \Pi^{\theta,r}\ket{j}\ket{l} \nonumber\\
&= \Exp_{\theta, r} \sum_b \sum_{x: r\cdot x = b} \bra{i} \ket{x^{\theta}} \bra{x^{\theta}} \ket{j} \cdot \bra{k} P^{\theta,r}_b \otimes Q^{\theta,r}_b \ket{l} \nonumber\\
&= \Exp_{\theta, r} \sum_b \frac{1}{2^{|H|}} \cdot \sum_{\substack{x_C = i_C = j_C \\ x_H: \,\, r_H \cdot x_H = r_C \cdot i_C \oplus b}} (-1)^{x_H \cdot (i_H \oplus j_H)} \cdot \bra{k} P^{\theta,r}_b \otimes Q^{\theta,r}_b\ket{l} \nonumber \\
&= \Exp_{\theta, r} \sum_b \frac{1}{2^{|H|}} \cdot \delta_{i_C = j_C} \cdot \sum_{x_H: \,\, r_H \cdot x_H = r_C \cdot i_C \oplus b} (-1)^{x_H \cdot (i_H \oplus j_H)} \cdot \bra{k} P^{\theta,r}_b \otimes Q^{\theta,r}_b\ket{l} \label{eq:6}
\end{align}
We can simplify the above expression further thanks to the following lemma.
\begin{lem}
\label{lem:sum-XOR}
Let $N \in \mathbb{N}$. Let $r, u \in \{0,1\}^N$ with $r\neq 0^N$, and let $b \in \{0,1\}$. Then,
$$ \sum_{x \in \{0,1\}^N: \,\, r \cdot x = b} (-1)^{x\cdot u} = \begin{cases}
 2^{N-1} \textnormal{ if }  u = 0^N  \\
2^{N-1}(-1)^b \textnormal{ if } u = r \\
0 \quad \textnormal{ otherwise }
\end{cases}
$$
\end{lem}
\begin{proof}
The first two cases are easy to verify directly. We prove the last case, namely that $\sum_{x \in \{0,1\}^N: \,\, r \cdot x = b} (-1)^{x\cdot u} = 0$ if $u \notin \{0^N, r\}$. Let $v \in \{0,1\}^N$ be any string such that $r\cdot v = b$. Then, notice that the set $\{x: \,r\cdot x = b\}$ is equal to the set $\{x' + v: \,r\cdot x' =0\}$. Thus, 
\begin{align*}
\sum_{x \in \{0,1\}^N: \,\, r \cdot x = b} (-1)^{x\cdot u} &= \sum_{x' \in \{0,1\}^N: \,\, r \cdot x' = 0} (-1)^{(x'+v) \cdot u} \\
&=(-1)^{v\cdot u} \sum_{x' \in \{0,1\}^N: \,\, r \cdot x' = 0} (-1)^{x'\cdot u} \,. 
\end{align*}
It is straightforward to verify that the latter is non-zero if and only if $u$ is orthogonal to the subspace $S =\{x': r\cdot x' =0\}$. Notice that $S$ is $(N-1)$-dimensional, and so $S^{\perp}$ must be $1$-dimensional. Thus, it is exactly $S^{\perp} = \{0^N, r\}$, as desired.
\end{proof}
Before applying \cref{lem:sum-XOR} to \cref{eq:6}, it will be convenient for us ``restrict'' our analysis to the uniform distribution on $(\theta,r)$ conditioned on $(\theta, r) \in S$, where 
\begin{equation}
\label{eq:S}
S = \{(\theta,r): |H| \geq 1 \,\, \,\land  \,\,\,
r_{H} \neq 0^{|H|} \} \,.
\end{equation}
This is convenient because it helps ensure that the hypothesis of \cref{lem:sum-XOR} is satisfied, and leads to a simpler expression. Now, notice, that $$(\{0,1\}^n \times \{0,1\}^n) \setminus S = \{(\theta,r): |H| = 0 \,\, \lor \,\, (|H| \geq 1 \,\, \land \,\,r_H = 0^H) \} \,.$$
Moreover, 
$$\{(\theta,r): |H| = 0 \,\, \lor \,\, (|H| \geq 1 \,\, \land \,\,r_H = 0^H) \} = \sum_{k=1} \binom{n}{k} = 2^n \,.$$
Thus, 
$$\Pr_{\substack{\theta,r}}[(\theta, r) \in S] \geq \frac{2^{2n} - 2^n}{2^{2n}} = 1 - 2^{-n} \,.$$
So, the restriction to $(\theta,r) \in S$ has only a negligible effect: more precisely, going back to the expression \eqref{eq:popt} that we are trying to bound, notice that, for any state $\ket{v}_{\sfABC}$ (and any $\{ P^{\theta,r}_b\}$, $\{Q^{\theta,r}_b\}$),
\begin{align}
\bra{v} \Exp_{\theta,r} \Pi^{\theta, r} \ket{v}  &= \Pr_{\theta, r}[(\theta,r) \in S] \cdot \bra{v}\Exp_{\substack{\theta,r \\ : (\theta, r) \in S}} \Pi^{\theta, r} \ket{v} \nonumber\\
&\quad+ \Pr_{\theta,r}[(\theta, r) \notin S] \cdot \bra{v} \Exp_{\substack{\theta,r \\ : (\theta, r) \notin S}} \Pi^{\theta, r} \ket{v} \nonumber\\
&\leq \bra{v}\Exp_{\substack{\theta,r \\ : (\theta, r) \in S}} \Pi^{\theta, r} \ket{v} + 2^{-n} \label{eq:bound_1}
\end{align}
(where the last line is by the fact that $\Pi^{\theta,r}$ is a projection and so $\Exp_{\substack{\theta,r \\ : (\theta, r) \notin S}} \Pi^{\theta, r}$ has operator norm at most $1$). Similarly, we have 
\begin{equation}
\label{eq:216}
\bra{v} \Exp_{\theta,r} \Pi^{\theta, r} \ket{v} \geq \bra{v}\Exp_{\substack{\theta,r \\ : (\theta, r) \in S}} \Pi^{\theta, r} \ket{v} - 2^{-n} \,.
\end{equation}
Thus, from here on, we will focus on bounding 
$$ \max_{\ket{v}_{\sfABC}} \,\, \bra{v} \Exp_{\substack{\theta,r \\ : (\theta, r) \in S}} \Pi^{\theta, r} \ket{v}   \,.$$
As before, for most of the rest of the proof, we will fix some arbitrary families of projections $\{P_{b}^{\theta, r}\}$ on $\sfB$ and  $\{Q_b^{\theta, r}\}$ on $\sfC$. 

Now, define
$$\widetilde{\gamma}_{ijkl} =  \bra{i}_{\sfA} \bra{k}_{\sfB\sfC} \Exp_{\substack{\theta,r \\ : (\theta, r) \in S}} \Pi^{\theta, r} \ket{j}_{\sfA}\ket{l}_{\sfB\sfC} \,$$
(i.e.\ $\widetilde{\gamma}_{ijkl}$ is defined analogously as $\gamma_{ijkl}$ except that the expectation is taken over $(\theta,r) \in S$). With an analogous calculation as for $\gamma_{ijkl}$ (just replacing the distribution that we are taking the expectation over), we obtain
\begin{equation*}
\widetilde{\gamma}_{ijkl} = \Exp_{\substack{\theta,r \\ : (\theta, r) \in S}} \sum_b \frac{1}{2^{|H|}} \cdot \delta_{i_C = j_C} \cdot \sum_{x_H: \,\, r_H \cdot x_H = r_C \cdot i_C \oplus b} (-1)^{x_H \cdot (i_H \oplus j_H)} \cdot \bra{k} P^{\theta,r}_b \otimes Q^{\theta,r}_b\ket{l}\,.
\end{equation*}
Applying \cref{lem:sum-XOR}, we get
\begin{align}
\widetilde{\gamma}_{ijkl} &= \Exp_{\substack{\theta,r \\ : (\theta, r) \in S}} \sum_b \frac{1}{2^{|H|}} \cdot 2^{|H|-1} \cdot \delta_{i_C = j_C} \cdot \big(\delta_{i_H = j_H} + \delta_{i_H = j_H \oplus r_H} \cdot (-1)^{r_C \cdot i_C \oplus b} \big) \cdot \bra{k} P^{\theta,r}_b \otimes Q^{\theta,r}_b\ket{l} \nonumber\\
&= \frac12 \Exp_{\substack{\theta,r \\ : (\theta, r) \in S}} \sum_b \big(\delta_{i=j} + \delta_{i=j \oplus (0_Cr_H)} (-1)^{r_C \cdot i_C \oplus b}\big) \cdot \bra{k} P^{\theta,r}_b \otimes Q^{\theta,r}_b\ket{l} \nonumber\\
&= \frac12 \delta_{i = j} \cdot \bra{k} \Exp_{\substack{\theta,r \\ : (\theta, r) \in S}} \sum_b P^{\theta,r}_b \otimes Q^{\theta,r}_b\ket{l} \label{eq:line-6}\\
&\,\,\,\,\,+  \frac12 \Exp_{\substack{\theta,r \\ : (\theta, r) \in S}}\delta_{i=j \oplus (0_Cr_H)} \cdot (-1)^{r_C \cdot i_C} \cdot   \bra{k} \sum_b (-1)^b P^{\theta,r}_b \otimes Q^{\theta,r}_b\ket{l} \label{eq:line-7} \\
&= \widetilde{\gamma}_{ijkl}^{(1)} + \widetilde{\gamma}_{ijkl}^{(2)} \,, \nonumber
\end{align}
where $\widetilde{\gamma}_{ijkl}^{(1)}$ equals the expression in line \eqref{eq:line-6}, and $\widetilde{\gamma}_{ijkl}^{(2)}$ equals the expression in line \eqref{eq:line-7}. 

Now, let $\ket{v} = \sum_{i,k} \alpha_{ik} \ket{i}_{\sfA}\ket{k}_{\sfB \sfC}$. Then, notice that
\begin{equation}
\label{eq:sum}
\bra{v} \Exp_{\substack{\theta,r \\ : (\theta, r) \in S}} \Pi^{\theta,r} \ket{v}  = \sum_{i,k,j,l} \overline{\alpha_{ik}} \cdot \alpha_{jl} \cdot \widetilde{\gamma}_{ijkl}^{(1)} + \sum_{i,k,j,l} \overline{\alpha_{ik}} \cdot \alpha_{jl} \cdot \widetilde{\gamma}_{ijkl}^{(2)} \,.
\end{equation}

Crucially, observe now that 
$$\sum_{i,k,j,l} \overline{\alpha_{ik}} \cdot \alpha_{jl} \cdot \widetilde{\gamma}_{ijkl}^{(1)} = \bra{v} W^{(1)}\ket{v} \,,$$
where 
$$ W^{(1)} = \frac12 I \otimes \Exp_{\theta,r} \sum_b P^{\theta,r}_b \otimes Q^{\theta,r}_b \,. $$
Observe also that 
$$\sum_{i,k,j,l} \overline{\alpha_{ik}} \cdot \alpha_{jl} \cdot \widetilde{\gamma}_{ijkl}^{(2)} = \bra{v} W^{(2)}\ket{v} \,,$$
where 
\begin{equation}
\label{eq:w2}
W^{(2)} = \frac12 \sum_{i,j}\ket{i}\bra{j}  \otimes M_{i,j}\,,
\end{equation}
and 
\begin{equation}
\label{eq:mij}
M_{i,j} = \frac12 \Exp_{\substack{\theta,r \\ : (\theta, r) \in S}}\delta_{i=j \oplus (0_Cr_H)} \cdot (-1)^{r_C \cdot i_C} \cdot \sum_b (-1)^b P^{\theta,r}_b \otimes Q^{\theta,r}_b \,.
\end{equation}

Plugging these into \cref{eq:sum}, we have 
\begin{equation}
\bra{v} \Exp_{\substack{\theta,r \\ : (\theta, r) \in S}} \Pi^{\theta,r} \ket{v} = \bra{v} W^{(1)}\ket{v} + \bra{v} W^{(2)}\ket{v} \,. \label{eq:W}
\end{equation}
Now, recall that we wish to show that 
$$ p_{opt}(n) = \max_{\ket{v}} \bra{v} \Exp_{\theta,r} \Pi^{\theta,r} \ket{v} \leq \frac12 + \epsilon(n) \,,$$
for some exponentially small function $\epsilon$. By Equations \eqref{eq:bound_1}, \eqref{eq:216}, and \eqref{eq:W}, this is equivalent to showing $$\max_{\ket{v}} \,\bra{v} W^{(1)}\ket{v} + \bra{v} W^{(2)}\ket{v} \leq \frac12 + \epsilon(n) \,,$$
for some exponentially small function $\epsilon$.

Let $\ket{v}$ be an arbitrary vector. We can write $\ket{v} = \sum_i \ket{i}_{\sfA}\ket{v_i}_{\sfB\sfC}$, for some un-normalized vectors $\ket{v_i}$ (that are not necessarily orthogonal), such that $\sum_i \| \ket{v_i} \|^2 = 1$. Then, 
\begin{equation}
\label{eq:133}
\bra{v} W^{(1)} \ket{v} = \bra{v} \frac12 I \otimes \Exp_{\theta,r} \sum_b P^{\theta,r}_b \otimes Q^{\theta,r}_b\ket{v} =  \Exp_{\theta,r,b}\sum_i 
\bra{v_i} P^{\theta,r}_b \otimes Q^{\theta,r}_b \ket{v_i} \,.
\end{equation}
We will now derive a similar expression for $\bra{v} W^{(2)} \ket{v}$. Recall the expression for $W^{(2)}$ from \eqref{eq:w2} and \eqref{eq:mij}. For convenience, we rewrite $W^{(2)}$ as 
$$W^{(2)} = \sum_{i, \Delta \in \{0,1\}^n }\ket{i}\bra{i\oplus \Delta} \otimes \widetilde{M}_{i,\Delta} \,,$$
where 
$$\widetilde{M}_{i, \Delta} = M_{i, i\oplus \Delta}  = \frac12 \Exp_{\substack{\theta,r \\ : (\theta, r) \in S}} \delta_{\Delta =  0_Cr_H} \cdot (-1)^{r_C \cdot i_C}  \sum_b (-1)^b P^{\theta,r}_b \otimes Q^{\theta,r}_b \,.$$
We can further rewrite
\begin{align}
\label{eq:m-i-delta-1-1}
\widetilde{M}_{i, \Delta} &= \frac12 \sum_{\substack{\theta',r'\\: \Delta =0_Cr'_H}} \Pr_{\substack{\theta,r \\ : (\theta, r) \in S}}[\theta',r'] \cdot (-1)^{r_C \cdot i_C}  \sum_b (-1)^b P^{\theta,r}_b \otimes Q^{\theta,r}_b \,,
\end{align}
where we are using the shorthand notation $$\Pr_{\substack{\theta,r \\ : (\theta, r) \in S}}[\theta',r'] = \Pr_{\substack{\theta, r \\ : (\theta,r) \in S}}[ (\theta,r) = (\theta',r')] \,.$$
We will keep using similar notation from here on. 

For some fixed $\theta',r'$, we have the following:
\begin{align}
\Pr_{\substack{\theta,r \\ : (\theta, r) \in S}} [ \theta',r']  =& \frac{\Pr_{\theta,r} [(\theta,r) = (\theta',r') \,\land \, (\theta, r) \in S]}{\Pr_{\theta, r}[(\theta,r) \in S]}
\end{align}
As we have calculated earlier, $\Pr_{\theta,r}[(\theta,r) \in S] = 1 - 2^{-n}$.
Thus, we have
\begin{align}
\Pr_{\substack{\theta,r \\ : (\theta, r) \in S}} [\theta',r']  &= \frac{1}{1-2^{-n}} \cdot \Pr_{\theta,r} [(\theta,r) = (\theta',r') \,\land \, (\theta, r) \in S] \nonumber\\
&= \big(1+\delta(n)\big) \cdot \Pr_{\theta,r} [(\theta,r) = (\theta',r') \,\land \, (\theta, r) \in S]\,. \nonumber
\end{align}
where, for convenience of notation, we have defined $\delta(n) = \frac{2^{-n}}{1-2^{-n}}$.
Now, note that 
$$\Pr_{\theta,r} [(\theta,r) = (\theta',r') \,\land \, (\theta, r) \in S] = 
\begin{cases}
\Pr_{\theta, r} [\theta',r'] \quad \textnormal{ if } (\theta',r') \in S \\
0 \quad \textnormal{ if } (\theta',r') \notin S
\end{cases}
$$
So,
\begin{align}
\Pr_{\substack{\theta,r \\ : (\theta, r) \in S}} [\theta',r']  &= 
\begin{cases}
\big(1+\delta(n)\big) \cdot \Pr_{\theta, r} [\theta',r'] \quad \textnormal{ if } (\theta',r') \in S \\
0 \quad \textnormal{ if } (\theta',r') \notin S
\end{cases} \nonumber \\
&=\begin{cases}
\big(1+\delta(n)\big) \cdot 2^{-2n} \quad \textnormal{ if } (\theta',r') \in S \\
0 \quad \textnormal{ if } (\theta',r') \notin S
\end{cases} \nonumber
\end{align}
Substituting the latter into \cref{eq:m-i-delta-1-1} and simplifying (as well as dropping the prime superscripts on $\theta'$ and $r'$), we obtain 
\begin{equation}
\label{eq:m-i-delta-2-2}
\widetilde{M}_{i, \Delta} = \frac12 \cdot 2^{-2n} \cdot \big(1+\delta(n)\big) \sum_{\substack{\theta,r\\: \Delta =0_Cr_H}} \sum_b (-1)^{r_C \cdot i_C + b} \,P^{\theta,r}_b \otimes Q^{\theta,r}_b  \,. \nonumber
\end{equation}
Now, again, consider an arbitrary vector $\ket{v} = \sum_i \ket{i}_{\sfA}\ket{v_i}_{\sfB\sfC}$, for some un-normalized vectors $\ket{v_i}$. We have 
\begin{align}
\bra{v} W^{(2)} \ket{v} &= \sum_{i, \Delta} \bra{v_{i}} \widetilde{M}_{i, \Delta} \ket{v_{i \oplus \Delta}} \nonumber\\
&= \frac12 \cdot 2^{-2n} \cdot \big(1+\delta(n)\big) \sum_{i, \Delta}\sum_{\substack{\theta,r\\: \Delta =0_Cr_H}} \sum_b (-1)^{r_C \cdot i_C + b} \bra{v_{i}} P^{\theta,r}_b \otimes Q^{\theta,r}_b \ket{v_{i\oplus \Delta}} \nonumber\\
&= \frac12 \cdot 2^{-2n} \cdot \big(1+\delta(n)\big) \sum_{i, \theta, r, b} (-1)^{r_C \cdot i_C +b} \bra{v_{i}} P^{\theta,r}_b \otimes Q^{\theta,r}_b \ket{v_{i\oplus (0_Cr_H)}} \nonumber \\
&=\big(1+\delta(n)\big) \cdot  \mathbb{E}_{\theta,r,b} \sum_i (-1)^{r_C \cdot i_C +b} \bra{v_{i}} P^{\theta,r}_b \otimes Q^{\theta,r}_b \ket{v_{i\oplus (0_Cr_H)}} \nonumber \\
&\leq \mathbb{E}_{\theta,r,b} \sum_i (-1)^{r_C \cdot i_C +b} \bra{v_{i}} P^{\theta,r}_b \otimes Q^{\theta,r}_b \ket{v_{i\oplus (0_Cr_H)}} + \delta(n) \,, \label{eq:1666}
\end{align}
where the last line follows by a triangle inequality and Cauchy-Schwarz (and the fact that $\sum_i \|\ket{v_i}\|^2 = 1$). By an analogous calculation, we also have the lower bound 
\begin{equation}
\label{eq:177}
\bra{v} W^{(2)} \ket{v} \geq \mathbb{E}_{\theta,r,b} \sum_i (-1)^{r_C \cdot i_C +b} \bra{v_{i}} P^{\theta,r}_b \otimes Q^{\theta,r}_b \ket{v_{i\oplus (0_Cr_H)}} - \delta(n)
\end{equation}

Thus, putting \eqref{eq:133}, \eqref{eq:1666}, and \eqref{eq:177} together, we have that the problem of showing 
$$p_{opt}(n) \leq \frac12 + \epsilon(n) \,,$$
for some exponentially small function $\epsilon$, is equivalent to proving Conjecture~\ref{conj:1}, stated in the next subsection.

\subsection{A conjecture implying an information-theoretically ``unclonable bit''}
\label{sec:conj_imply_UE}
By combining \eqref{eq:133}, \eqref{eq:1666}, and \eqref{eq:177}, we have that the following conjecture being true is equivalent to the optimal winning probability in the game $\mathsf{GLMonogamy}(n)$ (described at the start of Section~\ref{sec:1}) being exponentially close to $\frac12$ for fully quantum strategies.
\begin{conjecture}
\label{conj:1}
There exists an exponentially small function $\epsilon: \mathbb{N} \rightarrow [0,1]$ such that the following holds. For all $n \in \mathbb{N}$, for all families of projections $\{P^{\theta,r}_b\}_{\theta,r \in \{0,1\}^n, b \in \{0,1\}}$ and $\{Q^{\theta,r}_b\}_{\theta,r \in \{0,1\}^n, b \in \{0,1\}}$, respectively on registers $\sfB$ and $\sfC$ of arbitrary dimension, such that $P^{\theta, r}_0 + P^{\theta,r}_1 = I$ and $Q^{\theta, r}_0 + Q^{\theta, r}_1 = I$ for all $\theta,r$, and for all vectors $\ket{v_i}$,  $i \in \{0,1\}^n$, on $\sfB \sfC$, such that $\sum_i \|\ket{v_i}\|^2 = 1$,
$$  \Exp_{\theta,r,b}\sum_i 
\bra{v_i} P^{\theta,r}_b \otimes Q^{\theta,r}_b \ket{v_i} + \,\mathbb{E}_{\theta,r,b} \sum_i (-1)^{r_C \cdot i_C +b} \bra{v_{i}} P^{\theta,r}_b \otimes Q^{\theta,r}_b \ket{v_{i\oplus (0_Cr_H)}} \leq \frac12 + \epsilon(n) \,,$$
where, for a given $\theta$, we let $C = \{i: \theta_i = 0\}$ and $H = \{i: \theta_i = 1\}$.
\end{conjecture}

In the next section, we prove this conjecture for ``semi-classical'' strategies (recall that we defined ``semi-classical'' strategies in \cref{sec:1}).

\subsection{Continuing the analysis for semi-classical strategies}
\label{sec:analysis-unentangled}
To prove Conjecture \ref{conj:1} for semi-classical strategies, we proceed in two steps.
\begin{itemize}
    \item[1)] We will first show that $$\max_{\ket{v_i}: \sum_i \| \ket{v_i}\|^2 = 1} \,\,\Exp_{\theta,r,b}\sum_i 
\bra{v_i} P^{\theta,r}_b \otimes Q^{\theta,r}_b \ket{v_i} \leq \frac12$$.
    \item[2)] We will then show that $$ \max_{\ket{v_i}: \sum_i \| \ket{v_i}\|^2 = 1} \,\,\mathbb{E}_{\theta,r,b} \sum_i (-1)^{r_C \cdot i_C +b} \bra{v_{i}} P^{\theta,r}_b \otimes Q^{\theta,r}_b \ket{v_{i\oplus (0_Cr_H)}}$$ is exponentially small in $n$.
\end{itemize}
Step 1 is straightforward, and does not use the hypothesis that the strategy is ``semi-classical''. We have
\begin{align}
 \max_{\ket{v_i}: \sum_i \| \ket{v_i}\|^2 = 1} \,\,\Exp_{\theta,r,b}\sum_i 
\bra{v_i} P^{\theta,r}_b \otimes Q^{\theta,r}_b \ket{v_i} 
&= \frac12 \Exp_{\theta,r} \bra{v_i} \sum_b P^{\theta,r}_b \otimes Q^{\theta,r}_b \ket{v_i}  \\
&\leq \frac12 \Exp_{\theta,r} \bra{v_i} \sum_b P^{\theta,r}_b \otimes Q^{\theta,r}_b \ket{v_i} \\
&\leq \frac12 \Exp_{\theta,r} \sum_i \| \ket{v_i}\|^2 \\ 
&\leq \frac12 \,, 
\label{eq:final-1}
\end{align}
where the second inequality follows from the fact that $\sum_b P^{\theta,r}_b \otimes Q^{\theta,r}_b$ is a projection, along with one application of Cauchy-Schwarz.

We now proceed to step 2. From \cref{eq:177}, recall that
$$ \max_{\ket{v_i}: \sum_i \| \ket{v_i}\|^2 = 1} \,\,\mathbb{E}_{\theta,r,b} \sum_i (-1)^{r_C \cdot i_C +b} \bra{v_{i}} P^{\theta,r}_b \otimes Q^{\theta,r}_b \ket{v_{i\oplus (0_Cr_H)}}  \leq \max_{\ket{v}} \bra{v} W^{(2)} \ket{v} + \delta(n) \,,$$
where
$$W^{(2)} = \sum_{i, \Delta \in \{0,1\}^n }\ket{i}\bra{i\oplus \Delta} \otimes \widetilde{M}_{i,\Delta} \,,$$
$$
\widetilde{M}_{i, \Delta} = \frac12 \cdot 2^{-2n} \cdot \big(1+\delta(n)\big) \sum_{\substack{\theta,r\\: \Delta =0_Cr_H}} \sum_b (-1)^{r_C \cdot i_C + b} \,P^{\theta,r}_b \otimes Q^{\theta,r}_b  \,,
$$
and $\delta(n) = \frac{2^{-n}}{1-2^{-n}}$.

We will now use the hypothesis that the strategy is semi-classical. Thus, there exists a function $c$ such that, for all $\theta,r$, $$P_b^{\theta,r} = Q_b^{\theta,r} = \delta_{b = c(\theta,r)} \cdot I\,.$$
Thus, we can write
\begin{equation}
\label{eq:m-i-delta-2}
\widetilde{M}_{i, \Delta} = \frac12 \cdot 2^{-2n} \cdot \big(1+\delta(n)\big) \sum_{\substack{\theta,r\\: \Delta =0_Cr_H}} (-1)^{r_C \cdot i_C+c(\theta,r)} I \otimes I \,.
\end{equation}
 Notice that $W^{(2)}$ is Hermitian, so $\max_{\ket{v}} \bra{v}  W^{(2)} \ket{v} = \|W^{(2)}\|$. We will aim to bound $\| W^{(2)}\|$ by leveraging the following lemma.
\begin{lem}
\label{lem:operator-bound-1}
    Let $A = \sum_{i,j} \ket{i}\bra{j} \otimes M_{ij}$ where the $M_{ij}$ are any linear operators, and the $\ket{i}$ vectors form an orthonormal basis. Then, 
    $$\|A \| \leq \Big(\sum_{i,j} \|M_{ij} \|^2\Big)^{\frac12}$$
    (where $\| A\|$ is the operator norm of $A$).
\end{lem}
\begin{proof}
Recall that $\| A\| = \max_{\ket{v}: \|\ket{v}\| = 1} \|A \ket{v} \|$. This implies that $\|A\|^2 = \max_{\ket{v}: \| v\| = 1} \| A \ket{v} \|^2$. Now take any unit norm $\ket{v}$. We can write it as $\ket{v} = \sum_{j} \ket{j} \otimes \ket{v_j}$ for some (un-normalized) vectors $\ket{v_j}$ such that $\sum_j \|\ket{v_j}\|^2 =1$. Then, 
\begin{align}
\|A \ket{v} \|^2 &= \Big\| \sum_{i} \sum_{j} \ket{i} \otimes M_{ij} \ket{v_j} \Big\|^2 \nonumber\\
&= \Big\| \sum_{i}  \ket{i} \otimes \Big( \sum_j M_{ij} \ket{v_j} \Big) \Big\|^2  \nonumber\\
&= \sum_i \Big\| \sum_j M_{ij} \ket{v_j} \Big\|^2 \nonumber\\
&\leq \sum_i \big(\sum_j \| M_{ij} \ket{v_j} \| \big)^2 \nonumber\\
&\leq \sum_i \big( \sum_j \|M_{ij} \| \cdot \| \ket{v_j} \| \big)^2 \nonumber\\
&\leq \sum_i \big( \sum_j \| M_{ij}\|^2 \big) \cdot \big(\sum_j \|\ket{v_j} \|^2 \big) \nonumber\\
&= \sum_{ij} \| M_{ij}\|^2 \,, \nonumber
\end{align}
where the first inequality follows by a triangle inequality; the second inequality follows from the fact that $ \|M v\| \leq \|M \| \cdot \| \ket{v}\|$ for all operators $M$ and vectors $\ket{v}$; the third inequality follows by an application of Cauchy-Schwarz; and the last equality follows from the fact that $\sum_j \|\ket{v_j}\|^2 =1$.
\end{proof}

Applying \cref{lem:operator-bound-1}, we have 
\begin{equation}
\label{eq:155}
\|W^{(2)}\| \leq \Big( \sum_{i, \Delta} \|\widetilde{M}_{i, \Delta} \|^2 \Big)^{\frac12} \,.
\end{equation}
Now, fix some $i$ and $\Delta$. Then, from \cref{eq:m-i-delta-2}, we have
\begin{align}
\| \widetilde{M}_{i, \Delta} \| &= \frac12 \cdot 2^{-2n} \cdot \big(1+\delta(n)\big) \cdot \left |  \sum_{\substack{\theta,r\\: \Delta =0_Cr_H}} (-1)^{r_C \cdot i_C+c(\theta,r)} \right | \nonumber\\
&= \frac12 \cdot 2^{-2n} \cdot \big(1+\delta(n)\big) \cdot \left | \sum_{\theta: \Delta_C = 0_C}  \sum_{r_C} (-1)^{r_C \cdot i_C + c(\theta, r_C\Delta_H)}  \right | \nonumber\\
&\leq \frac12 \cdot 2^{-2n} \cdot \big(1+\delta(n)\big)\sum_{\theta: \Delta_C = 0_C} \left | \sum_{r_C} (-1)^{r_C \cdot i_C + c(\theta, r_C\Delta_H)} \right | \nonumber \,,
\end{align}
where the last line is by a triangle inequality. 

Now, combining the latter with \cref{eq:155}, we have
\begin{align}
\|W^{(2)}\|^2 &\leq \sum_{i, \Delta} \|\widetilde{M}_{i, \Delta} \|^2 \nonumber\\
&\leq  \frac14 \cdot 2^{-4n} \cdot \big(1+\delta(n)\big)^2 \sum_{i,\Delta}\left (  \sum_{\theta: \Delta_C = 0_C} \bigg |  \sum_{r_C} (-1)^{r_C \cdot i_C + c(\theta, r_C\Delta_H)} \bigg| \right )^2 \nonumber\\
&=  \frac14 \cdot 2^{-4n} \cdot \big(1+\delta(n)\big)^2 \sum_{\Delta} \sum_{\substack{\theta: \Delta_C = 0_C \\ \theta': \Delta_{C'} = 0_{C'}}} \sum_i \bigg|  \sum_{r_C} (-1)^{r_C \cdot i_C + c(\theta, r_C\Delta_H)}  \bigg| \cdot \bigg|  \sum_{r'_{C'}} (-1)^{r'_{C'} \cdot i_{C'} + c(\theta', r'_{C'}\Delta_{H'})} \bigg| \,, \nonumber
\end{align}
where in the latter expression, we used $C'$ and $H'$ to denote the subset of indices where $\theta'$ is respectively $0$ and $1$. Now, by Cauchy-Schwarz, we have
\begin{align}
\|W^{(2)}\|^2 &\leq \frac14 \cdot 2^{-4n} \cdot \big(1+\delta(n)\big)^2 \sum_{\Delta} \sum_{\substack{\theta: \Delta_C = 0_C \\ \theta': \Delta_{C'} = 0_{C'}}} \sqrt{\sum_i \bigg(  \sum_{r_C} (-1)^{r_C \cdot i_C + c(\theta, r_C\Delta_H)} \bigg)^2} \nonumber \\
&\quad\quad \quad \quad \quad\quad \quad \quad \quad \quad \quad \quad \quad \quad \quad \quad \quad \cdot \sqrt{\sum_{i'} \bigg(  \sum_{r'_{C'}} (-1)^{r'_{C'} \cdot i'_{C'} + c(\theta', r'_{C'}\Delta_{H'})} \bigg)^2} \nonumber\\
&= \frac14 \cdot 2^{-4n} \cdot \big(1+\delta(n)\big)^2 \sum_{\Delta} \sum_{\substack{\theta: \Delta_C = 0_C \\ \theta': \Delta_{C'} = 0_{C'}}} \sqrt{\sum_{i_{H}} \sum_{i_{C}} \bigg(  \sum_{r_C} (-1)^{r_C \cdot i_C + c(\theta, r_C\Delta_H)} \bigg)^2} \nonumber \\
&\quad\quad \quad \quad \quad\quad \quad \quad \quad \quad \quad \quad \quad \quad \quad \quad \quad \cdot \sqrt{\sum_{i'_{H'}} \sum_{i'_{C'}} \bigg(  \sum_{r'_{C'}} (-1)^{r'_{C'} \cdot i'_{C'} + c(\theta', r'_{C'}\Delta_{H'})} \bigg)^2} \,, \label{eq:166}
\end{align}
where in the last equality we just broke up the two sums over $i$ into sums over $i_H$ and $i_C$, and over $i'_{H'}$ and $i'_{C'}$ respectively, as this will help us shortly. 

The crucial step now is to invoke Parseval's identity (which follows immediately from the unitarity of the Fourier transform), in the following form.
\begin{lem}[Parseval's identity]
\label{lem:parseval}
Let $N \in \mathbb{N}$. Let $f: \{0,1\}^N \rightarrow \mathbb{R}$. Then, 
$$\sum_{i \in \{0,1\}^N}  \bigg (
 \sum_{r \in \{0,1\}^N} (-1)^{i \cdot r} \cdot f(r) \bigg)^2 = 2^n \sum_{r \in \{0,1\}^N} f(r)^2 $$
\end{lem}
Continuing from \eqref{eq:166}, we now apply Lemma \ref{lem:parseval} twice: once for the sum over $i_C$ with $N = |C|$, and $f(r_C) = (-1)^{c(\theta, r_{C}\Delta_{H})}$, and once for the sum over $i'_{C'}$ with $N = |C'|$, and $f(r'_{C'}) = (-1)^{c(\theta', r'_{C'}\Delta_{H'})}$. We obtain
\begin{align}
\|W^{(2)}\|^2 &\leq \frac14 \cdot 2^{-4n} \cdot \big(1+\delta(n)\big)^2 \sum_{\Delta} \sum_{\substack{\theta: \Delta_C = 0_C \\ \theta': \Delta_{C'} = 0_{C'}}} \sqrt{\sum_{i_H} 2^{|C|} \sum_{r_C} 1} \cdot \sqrt{\sum_{i'_{H'}} 2^{|C'|} \sum_{r'_{C'}} 1} \nonumber \\
&= \frac14 \cdot 2^{-4n} \cdot \big(1+\delta(n)\big)^2 \sum_{\Delta} \sum_{\theta: \Delta_C = 0_C} \sqrt{2^{|H| + 2|C|}} \cdot \sum_{\theta': \Delta_{C'} = 0_{C'}} \sqrt{2^{|H'| + 2|C'|}} \nonumber \\
&= \frac14 \cdot 2^{-4n} \cdot \big(1+\delta(n)\big)^2 \sum_{k =0}^n \,\,\,\sum_{\Delta: |\overline{\Delta}| = k} \,\,\, \sum_{j=0}^k \sum_{\theta: \Delta_C = 0_C \,\, \land \,\, |\overline{\theta}| = j} \sqrt{2^{n-j+2j}} \sum_{j'=0}^k \sum_{\theta': \Delta_{C'} = 0_{C'} \,\, \land \,\, |\overline{\theta'}| = j'} \sqrt{2^{n-j'+2j'}} \nonumber \\
&= \frac14 \cdot 2^{-3n} \cdot \big(1+\delta(n)\big)^2 \sum_{k =0}^n \binom{n}{k} \sum_{j=0}^k \,\,\,\sum_{\theta: \Delta_C = 0_C \,\, \land \,\, |\overline{\theta}| = j} 2^{j/2} \sum_{j'=0}^k \,\,\,\sum_{\theta': \Delta_{C'} = 0_{C'} \,\, \land \,\, |\overline{\theta'}| = j'} 2^{j'/2} \nonumber  \\
&= \frac14 \cdot 2^{-3n} \cdot \big(1+\delta(n)\big)^2 \sum_{k =0}^n \binom{n}{k} \sum_{j=0}^k \binom{k}{j} \cdot 2^{j/2} \sum_{j'=0}^k \binom{k}{j'} \cdot 2^{j'/2} \nonumber 
\end{align}
Now, using the binomial identity $\sum_{l=0}^m \binom{m}{l} x^l = (x+1)^m$ for all $m \in \mathbb{N}$ and $x \in \mathbb{R}$, we can simplify the above further to obtain
\begin{align}
\|W^{(2)}\|^2 & \leq \frac14 \cdot 2^{-3n} \cdot \big(1+\delta(n)\big)^2 \sum_{k =0}^n \binom{n}{k} (\sqrt{2}+1)^{2k} \nonumber \\
&= \frac14 \cdot 2^{-3n} \cdot \big(1+\delta(n)\big)^2 \cdot \Big((\sqrt{2}+1)^2 + 1\Big)^n \nonumber \\
&\leq  \frac14 \cdot \big(1+\delta(n)\big)^2 \cdot \Big(\frac{6.83}{8}\Big)^n \,. \nonumber
\end{align}
Taking square roots on both sides, we have 
\begin{equation} 
\label{eq:final-2}
\|W^{(2)}\| \leq \frac12 \cdot \big(1+\delta(n)\big) \cdot \Big(\frac{6.83}{8}\Big)^{n/2} \leq  \frac12 \cdot \big(1+\delta(n)\big) \cdot 0.93^n \,.
\end{equation}
Let $p_{opt\textnormal{-}semi\textnormal{-}classical}(n)$ denote the optimal winning probability for semi-classical strategies. Putting everything together from Equations \eqref{eq:final-1} and \eqref{eq:final-2}, we have that, for semi-classical strategies,
$$p_{opt\textnormal{-}semi\textnormal{-}classical}(n) \leq \frac12 + \max_{\ket{v}} \bra{v} W^{(2)}\ket{v}  \leq \frac12 +  \frac12 \cdot \big(1+\delta(n)\big) \cdot 0.93^n \,.$$
Now, recall that we defined $\delta(n) = \frac{2^{-n}}{1-2^{-n}} \leq 1$, for all $n\geq 1$. Thus, for semi-classical strategies,
$$p_{opt\textnormal{-}semi\textnormal{-}classical}(n) \leq \frac12 + 0.93^n \,.$$


\bibliographystyle{alpha}
\bibliography{references}

\end{document}